\documentclass[twoside]{IEEEtran}
\usepackage{schoolbook}
\usepackage{amsmath}                
\usepackage{amssymb}  
\usepackage{psfrag,graphicx}
\usepackage{epsfig}

\newtheorem{thm}{Theorem}
\newtheorem{lem}{Lemma}

\newtheorem{rem}{Remark}

\newcommand{\proofheader}[1]{\medskip\noindent\textsc{#1}}
\newcommand{\E}[1]{\mathrm{E}(#1)}
\def\tr{\mathrm{tr}\,}

\newcommand{\Hardy}[1]{{\mathcal H_#1}}

\newcommand{\Lebesgue}[1]{{\mathcal L_#1}}
\newcommand{\Htwo}{{\Hardy{2}}}
\newcommand{\Hone}{{\Hardy{1}}}
\newcommand{\Hinf}{{\Hardy{\infty}}}
\newcommand{\Ltwo}{{\Lebesgue{2}}}
\newcommand{\Lone}{{\Lebesgue{1}}}
\newcommand{\Linf}{{\Lebesgue{\infty}}}

\newcommand{\Smirnov}{\mathcal{N}^+}

\newcommand{\re}{\operatorname{Re}}
\newcommand{\im}{\operatorname{Im}}

\newcommand{\minimize}{\operatornamewithlimits{minimize}}
\newcommand{\maximize}{\operatornamewithlimits{maximize}}
\newcommand{\eqdef}{\overset{\underset{\mathrm{def}}{}}{=}}

\newcommand{\T}{\mathbb{T}}

\newcommand{\half}{{\frac{1}{2}}}
\newcommand{\quarter}{{\frac{1}{4}}}

\newcommand{\norm}[1]{\left\|{#1}\right\|}
\newcommand{\eiw}{e^{i\omega}}
\newcommand{\intT}{\int_{-\pi}^{\pi}}

\begin{document}
\title{Optimal Linear Joint Source-Channel Coding with Delay Constraint}
\author{Erik~Johannesson, 
        Anders~Rantzer,~\IEEEmembership{Fellow,~IEEE,}
        Bo~Bernhardsson,
        and~Andrey~Ghulchak
\thanks{Submitted to IEEE Transactions on Information Theory on March 28th 2012.}%
\thanks{This work was supported by the Swedish Research Council through the Linnaeus Center LCCC; the European Union's Seventh Framework Programme under grant agreement number 224428, project acronym CHAT; and the ELLIIT Strategic Research Center. The material in this paper has been presented in part at the American Control Conference (ACC), Baltimore, MA, USA, June 2010, the 19th International Symposium on Mathematical Theory of Networks and Systems (MTNS), Budapest, Hungary, July 2010, and in the Ph.D. thesis \cite{joh11phd}.}
\thanks{The authors are with the Department
of Automatic Control, Lund University, Lund, Sweden.
E-mail: \{erik, rantzer, bob, \hbox{andrey}\}@control.lth.se.}%
\thanks{\copyright 2012 IEEE. Personal use of this material is permitted. Permission from IEEE must be obtained for all other uses, in any current or future media, including reprinting/republishing this material for advertising or promotional purposes, creating new collective works, for resale or redistribution to servers or lists, or reuse of any copyrighted component of this work in other works.}%
}

\maketitle

\begin{abstract}
The problem of joint source-channel coding is considered for a stationary remote (noisy) Gaussian source and a Gaussian channel. The encoder and decoder are assumed to be causal and their combined operations are subject to a delay constraint. It is shown that, under the mean-square error distortion metric, an optimal encoder-decoder pair from the linear and time-invariant (LTI) class can be found by minimization of a convex functional and a spectral factorization. The functional to be minimized is the sum of the well-known cost in a corresponding Wiener filter problem and a new term, which is induced by the channel noise and whose coefficient is the inverse of the channel's signal-to-noise ratio. This result is shown to also hold in the case of vector-valued signals, assuming parallel additive white Gaussian noise channels. It is also shown that optimal LTI encoders and decoders generally require infinite memory, which implies that approximations are necessary.
A numerical example is provided, which compares the performance to the lower bound provided by rate-distortion theory.
\end{abstract}


\IEEEpeerreviewmaketitle

\section{Introduction}
\IEEEPARstart{T}{he} design of systems for point-to-point communication of analog data over noisy communication channels has a theoretical basis in Shannon's separation theorem. The theorem gives a bound on the optimal performance theoretically achievable (OPTA) by any communication system.
Specifically, it says that the distortion can not be made smaller than $\mathcal{D}_{\text{min}}$, which can be obtained from
\begin{equation}\label{eq:minimum_distortion}
\mathcal{R}(\mathcal{D}_\text{min}) = \mathcal{C},
\end{equation}
where $\mathcal{R}(\mathcal{D})$ is the rate-distortion function, which is given by the source statistics and the distortion measure, and $\mathcal{C}$ is the channel capacity. Under appropriate assumptions, the separation theorem also shows that it is possible to come arbitrarily close to $\mathcal{D}_{\text{min}}$ by the combination of source coding and channel coding. These codes can, in principle, be independently developed without loss. This means that the channel code designer does not need to know anything about the source, and vice-versa, which is clearly a practical advantage.

The separation theorem does, however, rely on asymptotic arguments where the delay and the size of the codebook are allowed to increase indefinitely. Consequently, it does not hold in presence of delay or complexity constraints and imposing such constraints generally renders the distortion bound unachievable. 
Since infinitely large delays or codebooks are not possible in practice, a suboptimal performance may have  to be accepted. 
Moreover, to minimize the distortion in the presence of these constraints, it may be necessary to abandon the separation-based design and consider joint source-channel codes. 

This is the subject of the present paper, where we consider transmission of a stationary colored Gaussian source over a power-constrained channel with additive colored Gaussian noise, under the mean-square error (MSE) distortion criterion. The encoder and decoder are constrained to be causal and their combined operations are subject to a delay constraint. Further, we allow for the possibility of a remote (noisy) source. The situation is illustrated in Fig. \ref{fig:comm_system}.

\begin{figure}[tb]
  \centerline{\includegraphics[width=1\hsize]{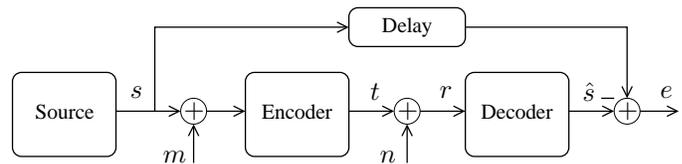}} 
  \caption{The encoder measures the source sequence $s$ plus the measurement noise $m$ and transmits $t$ over the channel. The decoder receives $t$ plus the channel noise $n$ and forms $\hat{s}$, the estimate of $s$. Each source element has to be estimated after a given delay in order to minimize the error $e$.}
  \label{fig:comm_system}
\end{figure}

The encoder and decoder will be restricted to the class of linear and time-invariant (LTI) filters. The linearity assumption and the additive noise models allow us to formulate the distortion minimization as a transfer function optimization problem. The main result is that a jointly optimal encoder-decoder pair from the LTI class can be found by first minimizing a functional of the form
\begin{equation}\label{eq:MainResultFunctional}
  \norm{R-X}_2^2 + \frac{1}{\sigma^2} \norm{XN}_1^2,
\end{equation}
where $R\in\Linf$ and $N\in\Hinf$ are given transfer functions and $\sigma^2$ is the signal-to-noise ratio (SNR), over $X\in\Htwo$. The encoder and decoder are then obtained from a spectral factorization. A corresponding result is also shown to hold in the case with vector-valued signals and parallel additive white Gaussian noise (AWGN) channels.

The restriction to linear encoders and decoders may obviously result in suboptimal solutions. Nevertheless, the linear solution to any problem instance will provide an upper bound to the minimum distortion possible for the given SNR, delay constraint, and signal spectra. Moreover, the proposed design methods are relatively simple and computationally feasible.

An application where this problem formulation could be relevant is the transmission of speech in mobile communication. The source signal to be estimated at the receiver is the speech signal. The delay constraint is based on the acceptable latency and the noise is any background sound present at the microphone. 

The rest of this section will present the relevant previous research and alternative interpretations of the problem. Section \ref{sec:notation} presents the mathematical notation used in this paper. The exact problem formulation is given in Section \ref{sec:problemformulation}. Section \ref{sec:solution} is devoted to the solution of the problem, first in the scalar and then in the vector case, followed by 
a theorem stating that optimal LTI encoders and decoders require infinite memory. Section \ref{sec:numerics} presents a procedure for numerical solution and a numerical example where the performance of the optimal LTI encoders and decoders is compared to the lower bound provided by the separation theorem. Finally, Section \ref{sec:conclusion} presents the conclusions and discusses further research. Some technical lemmas have been put in the appendix.

\subsection{Previous Research}
The problem studied in this paper is closely related to that of finding the optimal modulation matrices for linear coding and decoding of a Gaussian vector source for transmission over a Gaussian vector channel. Optimal modulation matrices were derived in \cite{pilc1969}, where it was also shown that linear modulation is only optimal when the source and channel can be matched. That is, when their dimensions match and the source and channel noise covariance matrices can be diagonalized into uniform variances. The same problem was considered in \cite{lee_petersen1976}, where the solution was also given for the case when the channel components have individual power constraints. The performance of optimal linear coding was compared, for a number of cases, to the OPTA, given by (\ref{eq:minimum_distortion}), in \cite{basar_sankur_abut1980}.\footnote{In all of these three papers, one may view the source vectors as vectors in a one-shot problem, where there is no dependence over time, or as finite sequences. In the former interpretation, the solution satisfies a zero-delay constraint, but this is not very interesting due to the lack of dependence. In the latter interpretation, a delay constraint would translate to requiring the matrices to be lower-triangular, which is not done.}

The general suboptimality of linear coding arises from the fact that it cannot match any colored Gaussian source to any colored Gaussian channel. It has recently been shown, however, that such matching can be achieved by the combination of prediction and modulo-lattice operations \cite{kochman_zamir2011}.

The problem of coding with a remote source was first considered for the Gaussian case with additive noise and MSE distortion in \cite{Dobrushin1962}. It was shown that the problem is asymptotically equivalent to, and can thus be reduced to, the fully observed case and that an optimal encoder generally has a structure consisting of an optimal estimator followed by optimal encoding for a noise-free source. 
This structural result was generalized to the non-gaussian and finite time horizon cases in \cite{Wolf1970}. 	
The problem was further studied in \cite{berger1971rate}, where it was noted that in the case of white source noise, the criterion in the reduced problem is given by the conditional expectation of the original criterion given the encoder input.
It was pointed out in \cite{Witsenhausen1980} that the equivalence in \cite{Dobrushin1962} actually was proved for the one-shot problem as well. Moreover, it was shown that the reduction to the non-remote problem follows from a general "disconnection principle".
In the literature, the problem of coding with a remote source often includes the possibility of noise at the receiver as well. The main motivation for excluding that possibility here is the fact, noted in \cite{Wolf1970}, that the optimality of an encoder-decoder design is independent of additive and independent zero-mean noise at the receiver.

Coding problems with delay constraints have not received the same level of attention as their classical counterparts. 
Some structural results have, however, been obtained. The optimal causal source coder for a white source has been found to be memoryless \cite{Neuhoff1982}.
For a Markov source of order $k$ and delay constraint $d$, an optimal real-time source coder only needs to use the last $\max \lbrace k,d+1\rbrace$ source symbols plus the current state of the decoder. No such memory bound is given, however, when the encoder does not have access to the decoder state \cite{Witsenhausen1979}. 
Joint source-channel coding with noiseless feedback was considered for finite alphabet sources in	\cite{Walrand} where it was demonstrated that feedback is useful in general, but that coding is useless for a class of channels with a certain symmetry property.
The results in \cite{Witsenhausen1979, Walrand} have been generalized in \cite{teneketzis2006}, which also gives a nice overview of the literature on real-time coding.
Conditions have also been found for when optimal performance can be achieved without coding (even when allowing coding systems with arbitrary delay) \cite{gastpar}. 

Since the OPTA given by (\ref{eq:minimum_distortion}) cannot generally be achieved in the presence of delay constraints, a relevant question to ask is of course what the OPTA is when there are such constraints. A partial answer in the form of upper bounds on the rate-distortion functions for zero-delay and causal source coding is given in the recent paper \cite{derpich2011}. Interestingly, some of the results in that paper are obtained by solving a problem which is somewhat similar to the one considered in this paper. The solution of that problem can be applied to solve some particular instances of the problem considered in this paper. The main difference is that they assume that the encoder has access to noiseless feedback from the channel output. Moreover, only the scalar case with zero delay constraint and no noise at the source is considered.  The same problem has previously been 
considered in \cite{Derpich08, DerpichPhd} as a means to design optimal scalar feedback quantization schemes. 

Real-time source coding for a remote source has been considered in \cite{Borkar2001}. The structural results of \cite{Witsenhausen1979, Walrand} were extended to cover remote sources in \cite{Yuksel2010}, which also presented a separation result for the linear-quadratic Gaussian case similar to the one in \cite{Dobrushin1962}. 
A method for design of optimal real-time coding systems for noisy channels was presented in \cite{Mahajan2008} using noisy feedback and in \cite{mahajan09} without feedback. However, there seems to be no method for efficient numerical application of the solution.

\subsection{Alternative Interpretations}
It is possible to make two alternative interpretations of the problem illustrated in Fig. \ref{fig:comm_system}.
\subsubsection{Connection to Wiener Filter}	
The problem of estimating a signal that is measured with additive noise under an MSE criterion is solved by the Wiener filter \cite{Wiener1964}. The filter is usually obtained by solving the Wiener-Hopf equations, but can also be expressed in the frequency domain as the stable filter $K$ that minimizes 
\begin{align}\label{eq:wienercost}
  \norm{(z^{-d}-K)S}_2^2 + \norm{KM}_2^2,
\end{align}
where $d$ is the allowed time delay and $S$ and $M$ are transfer functions that represent the frequency characteristics of the signal of interest and the measurement noise, respectively.

It is possible to interpret the problem in Fig. \ref{fig:comm_system} as a distributed Wiener filtering problem, where the filter is separated into two different locations. The communication channel is used to model the communication constraint between the two locations. This interpretation is strengthened by the fact that minimization of (\ref{eq:wienercost}) is equivalent to minimizing
\begin{equation}\label{eq:wienercost2}
  \norm{R-X}_2^2, 
\end{equation}
where $R$ is the same transfer function as in (\ref{eq:MainResultFunctional}), over $X\in\Htwo$. Comparing (\ref{eq:wienercost2}) with 
(\ref{eq:MainResultFunctional}) it is seen that the cost in the present problem is equal to the cost in a Wiener filtering problem plus an additional term, which is induced by the communication channel. Since the coefficient of the new term is the inverse of the channel's SNR, the cost is asymptotically equal to that in the Wiener filtering problem when the SNR tends to infinity.

\subsubsection{As a Feed-Forward Control Problem}
Fig. \ref{fig:comm_system} may be interpreted as follows: The source signal is a disturbance that will affect some system where a controller (the decoder) can compensate. The controller has a remote sensor that measures the disturbance and transmits information to the controller over the channel. In this interpretation the delay block may also include any dynamics that the disturbance passes through on the way. A similar interpretation was discussed in \cite{Witsenhausen1980}. 

A similar problem setup was studied in \cite{martins07}, where information theory was used to 
find a lower bound on the reduction of entropy rate made possible by side information communicated through a general channel with known capacity. Under stationarity assumptions, this
was used to derive a lower bound, which is a generalization of Bode's
integral equation, on a sensitivity-like function. 

\section{Notation} \label{sec:notation}
The techniques in this paper rely on concepts from functional analysis, such as $\Lebesgue{p}$ (Lebesgue), $\Hardy{p}$ (Hardy) and $\Smirnov$ (Smirnov) function classes and inner-outer factorizations. To conserve space, only some of the most important facts will be given here. The interested reader is referred either to \cite{joh11phd} or to \cite{Garnett}, \cite{rudin86real} and \cite{Inouye} for the remaining relevant definitions and theorems.

The natural logarithm is denoted $\log$.
The complex unit circle is denoted by $\T$.
The singular value decomposition of $A$ is taken as $A=U \Sigma V^*$, where $\Sigma$ is square. 
A singular value decomposition of a transfer matrix \mbox{$X \in \Lebesgue{p}$} is defined pointwise on $\T$ as \[X(\eiw) =U(\eiw) \Sigma (\eiw) V^* (\eiw),\] where $U,V \in \Linf$ and $\Sigma \in \Lebesgue{p}$.

For matrix-valued functions $X(z), Y(z)$ defined on $\T$, define 
\begin{equation*}
  \langle X, Y \rangle = \intT \tr\left(X(e^{i\omega})^* Y(e^{i\omega})\right) \frac{d\omega}{2\pi}
\end{equation*}
and the norms
\begin{align*}
  \norm{X}_1 &= \intT \sqrt{X(e^{i\omega})^*X(e^{i\omega})} \ \frac{d\omega}{2\pi} \\
  \norm{X}_2 &= \left(\intT \norm{X(e^{i\omega})}_F^2\ \frac{d\omega}{2\pi}\right)^{1/2}, 
\end{align*}
where $\norm{\cdot}_F$ is the Frobenius norm.

When a function in $\Hardy{p}$ is evaluated on $\T$, it is to be understood as the radial limit $\lim_{r\rightarrow 1^+} X(rz)$. 
The arguments of transfer matrices will often be omitted when they are clear from the context.
Equalities and inequalities involving functions evaluated on $\T$ are to be interpreted as holding almost everywhere
on $\T$. 

\section{Problem Formulation}  \label{sec:problemformulation}
Consider the system in Fig. \ref{fig:comm_system}. The source $s$, source noise $m$ and channel noise $n$ are assumed to be mutually independent, stationary Gaussian\footnote{Since only linear solutions are considered, it does not matter if the source, measurement noise or the channel noise are Gaussian or not. Linear solutions may, of course, be more or less suboptimal depending on the distributions.} sequences with zero mean and known covariance functions. The communication channel has additive noise and a power constraint. That is,
\begin{gather}
r=t+n \\ \E{t(k)^2} \leq \sigma^2. \label{eq:scalar_power_constraint_time_domain}
\end{gather}
Denote the encoder mapping by $\gamma(\cdot)$ and the decoder mapping by $\delta(\cdot)$. The encoder and decoder are assumed to be causal LTI filters with inputs $s+m$ and $r$, respectively. The estimate of the source sequence is
\begin{equation}
\hat{s} = \delta(t+n) = \delta(\gamma(s+m)+n).
\end{equation} 
Denoting the delay, in number of samples, by $d$, the reconstruction error is 
\begin{equation}
e(k) = s(k-d)-\hat{s}(k).
\end{equation}
The objective is to choose the encoder and decoder to minimize the stationary value of the MSE, or $\E{e(k)^2}$, subject to the power constraint.

Due to the linearity assumption, the problem can be formulated in the frequency domain, as is illustrated in Fig. \ref{fig:coding_block_diagram}. In this formulation, all the inputs are mutually independent, zero mean, white noise sequences with unit variance. The transfer functions $S(z), M(z)$ and $N(z)$ are spectral factors of the sequences $s, m$ and $n$, respectively. The encoder and the decoder are represented by the transfer functions $C(z)$ and $D(z)$. In this formulation, the problem has been generalized in two aspects:
\begin{itemize}
\item The delay is replaced by a general LTI filter $P$. That is, the objective is to estimate the source signal after it has passed through $P$.
\item The error $e$ is passed through a LTI filter $W$, representing a frequency weighting function, before minimization.
\end{itemize}

\begin{figure}[tb]
  \centerline{\includegraphics[width=1\hsize]{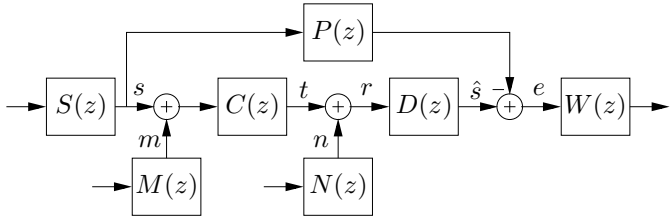}} 
  \caption{A representation of the problem in the frequency domain. The transfer functions $S,M$ and $N$ are spectral factors of the source, measurement noise and channel noise, respectively. The delay constraint is determined by $P$. The encoder and decoder filters are given by $C$ and $D$. $W$ is an optional frequency weight.}
  \label{fig:coding_block_diagram}
\end{figure}

It is assumed that $S,M,N,P,W \in \Hinf$, that $N,W$ are invertible in $\Hinf$ and that
\begin{align}\label{eq:FGepsilon}
  \exists \varepsilon > 0 \text{ such that } SS^*+MM^*\geq \varepsilon \text{ on }\T,
\end{align}
which implies that $S$ and $M$ have no common zeros on the unit circle (an equivalent condition if $S(z)$ and $M(z)$ are rational functions). 

The objective is to choose $C$ and $D$ to minimize the stationary variance of $e$ after filtering by $W$. By expressing the z-transform of $e$ in terms of the transfer functions in Fig. \ref{fig:coding_block_diagram}, this quantity can be expressed as 
\begin{equation}\label{eq:coding_SISO_objective}
  J(C,D) = \norm{W(P-DC)S}_2^2 + \norm{WDCM}_2^2 + \norm{WDN}_2^2.
\end{equation} 
Similarly, the power constraint on $t$ can be written as
\begin{equation}\label{eq:coding_constraint}
  \norm{CS}_2^2 + \norm{CM}_2^2 \leq \sigma^2.
\end{equation}
It follows from (\ref{eq:coding_SISO_objective}) and (\ref{eq:coding_constraint}) that $C$ and $D$ need to be square integrable on the unit circle in order for $J(C,D)$ to be finite and the power constraint to be satisfied. Since the encoder and decoder also should be causal and stable this implies that the optimization should be performed over $C,D\in\Htwo$.

\section{Optimal Linear Encoder and Decoder}\label{sec:solution}
The problem of finding an optimal linear encoder and decoder will first be solved in the scalar case. The solution will then, under some additional assumptions, be generalized to the vector case.
\subsection{Scalar case}
The objective function $J(C,D)$ is clearly not convex in the pair $(C,D)$ due to the appearance of the product $DC$. 
In order to find a minimum, the optimization problem will be solved in two steps.

The idea is to first consider the product $DC$ as given and then to find an optimal factorization of this product. The factorization gives an analytical expression for the cost in terms of the product, which means that optimization of the objective may then be performed over the product. When an optimal product is found, the optimality conditions from the solution to the factorization problem can then be applied to find optimal $C$ and $D$.

First, however, it will be shown that the power constraint (\ref{eq:coding_constraint}) can be equivalently written as
\begin{equation}\label{eq:coding_constraint_in_H}
  \norm{CH}_2^2 \leq \sigma^2,
\end{equation}
where the function $H$ has some nice properties.
\begin{lem}\label{lem:Hdefine}
  Suppose that $S, M \in \Hinf$ and that (\ref{eq:FGepsilon}) holds. Then there exists $H \in \Hinf$ with $H^{-1}\in\Hinf$ such that 
  \begin{equation}\label{eq:HFG}
    HH^* = SS^* + MM^* \text{ on } \T.
  \end{equation}
\end{lem}
\begin{proof}
  By (\ref{eq:FGepsilon}) and the factorization theorem in \cite{Wiener1959} there exists an outer function $H\in\Htwo$ such that (\ref{eq:HFG}) holds.
  Since $S,M\in\Hinf$ it follows that $H\in\Hinf$. Moreover, it follows from (\ref{eq:FGepsilon}) that $\norm{H^{-1}}_\infty \leq 1/\sqrt{\varepsilon}$ and since $H$ is outer it then follows from Lemma \ref{lem:product_with_inverse_in_Hp} (in the appendix) 
   that $H^{-1}\in\Hinf$. 
\end{proof}

Now, introduce $K=DC\in\Hone$. The objective (\ref{eq:coding_SISO_objective}) can then be written as 
\begin{equation}\label{eq:coding_SISO_objective_with_K}
\norm{W(P-K)S}_2^2 + \norm{WKM}_2^2 + \norm{WDN}_2^2.
\end{equation}
Note that the first two terms are constant for fixed $K$. The minimum over $C$ and $D$, given $K$, is thus obtained by minimizing the third term in (\ref{eq:coding_SISO_objective_with_K}) subject to (\ref{eq:coding_constraint_in_H}) and $K=DC$. This minimization problem is called the \emph{optimal factorization problem}. 

The interpretation is that for any given product of the encoder and decoder, the contribution to the objective of the signals that pass through both the encoder and the decoder is not affected by the choice of the factors $C$ and $D$ --- only their product matters. The channel noise, however, only passes through the decoder, which means that $D$ (and implicitly $C$ since $C=D^{-1}K$) should be chosen to minimize the impact of the channel noise on the objective. 
The solution to the scalar version of the optimal factorization problem is given by the following lemma.
\begin{lem}[Optimal factorization, scalar case] \label{lem:factorization_coding}
  Suppose that $\sigma > 0$, $K\in \Hone$ and that
  $H,N,W\in\Hinf$ are invertible in $\Hinf$. Then the optimization problem
  \begin{equation} \label{intermediateobjective} 
    \minimize_{C,D\in\Htwo} \norm{WDN}_2^2
  \end{equation} 
  subject to
  \begin{equation}\label{lemmaconstraints}
    K=DC,\quad \norm{CH}_2^2 \leq \sigma^2
  \end{equation}
  attains the minimum value
  \begin{align}\label{minval}
    \frac{1}{\sigma^2} \norm{WKHN}_1^2.
  \end{align}

  Moreover, if $K$ is not identically zero then $C,D\in\Htwo$ are optimal if and only if $DC=K$ and
  \begin{equation}\label{optimalitycondition}
    |C|^2 = \frac{\sigma^2}{\norm{WKHN}_1}  \left|\frac{WKN}{H}\right| \text{ on } \T.
  \end{equation}
  If $K=0$, then the minimum is achieved by $D=0$
  and any function $C\in\Htwo$ that satisfies $\norm{CH}_2^2 \leq \sigma^2$.
\end{lem}
\begin{proof}
  If $K=0$ the proof is trivial, so assume that $K$ is not identically
  zero. Then $C$ is not identically zero and $D=KC^{-1}$.
  Then (\ref{lemmaconstraints}) and Cauchy-Schwarz's inequality gives
  \begin{align*}
    \norm{WDN}_2^2 &= \norm{WKC^{-1}N}_2^2 
    \geq \frac{\norm{CH}_2^2}{\sigma^2} \norm{WKC^{-1}N}_2^2 
    \\&\geq  \frac{1}{\sigma^2} \left\langle \left |CH\right |,\left |WKC^{-1}N\right | \right\rangle^2
    =\frac{1}{\sigma^2} \norm{WKHN}_1^2 
  \end{align*}
  This shows that (\ref{minval}) is a lower bound on
  the value. Equality holds if and only if
  $|WKC^{-1}N|$ and $|CH|$ are proportional on $\T$ and $\norm{CH}_2^2= \sigma^2$.
  It is easily verified that this is equivalent to (\ref{optimalitycondition}).  
  Thus, $C$ and $D$ achieve the lower bound if and only if
  $D=KC^{-1}$ and (\ref{optimalitycondition})
  holds.

  It remains to show existence of such $C,D\in\Htwo$. Note that
  $WKNH^{-1} \in \Hone$ is not identically zero. 
  Hence, by Theorem 17.17 in \cite{rudin86real},
  $\log{|WKNH^{-1}|} \in \Lone$.
  It follows from the factorization theorem in \cite{Wiener1959} that there exists an outer $C\in \Htwo$ that satisfies (\ref{optimalitycondition}). Thus $$\norm{KC^{-1}}_2^2 = \frac{1}{\sigma^2}\norm{WKHN}_1\norm{W^{-1}KHN^{-1}}_1 < \infty,$$ 
  so $D=KC^{-1}\in\Ltwo$.
  Since $K\in\Hone$ and $C\in\Htwo$ is
  outer it follows from Lemma~\ref{lem:product_with_inverse_in_Hp} (in the appendix) that $D=KC^{-1}\in\Htwo$.
\end{proof}

\begin{rem}
  Optimal $D$ satisfy 
  \begin{equation}\label{eq:coding_optimality_conition_D}
    |D|^2 = \frac{\norm{WKHN}_1}{\sigma^2} \left|\frac{KH}{WN}\right| \text{ on } \T.
  \end{equation}
  Apparently, the magnitudes of $C$ and $D$ are both proportional to the square root of the magnitude of $K$. This provides some intuition to why the minimum value depends on the 1-norm of $K$.
\end{rem}

\begin{rem}\label{rem:coding_CD_nonunique}
  The existence part of Lemma \ref{lem:factorization_coding} shows that a particular solution, where $C$ is outer, can be obtained. By using the freedom available in spectral factorization, it is possible to obtain other solutions, for example by changing the sign of both $C$ and $D$, or by instead choosing $D$ to be outer. More generally, in the rational case, any non-minimum phase zeros or time delays could be located in $C$ or $D$.
\end{rem}

For any given $K$ an optimal encoder-decoder pair, under the constraint that their product is $K$, is specified by (\ref{optimalitycondition}) and (\ref{eq:coding_optimality_conition_D}), respectively. An optimal $K$ can in turn be obtained by inserting the minimum value of $\norm{WDN}_2^2$ into (\ref{eq:coding_SISO_objective_with_K}) and minimizing
\begin{align*}
  \varphi(K) &= \norm{W(P-K)S}_2^2 + \norm{WKM}_2^2 \\&+ \frac{1}{\sigma^2} \norm{WK \begin{bmatrix}S & M\end{bmatrix}
  N}_1^2
\end{align*}
over $K$. This is a convex problem. That this procedure in fact solves the main problem is shown by the following theorem, which is the main result of this paper.
\begin{thm}\label{thm:coding_equivalence_SISO}
  Suppose that $M, N, S, P, W \in \Hinf$, where $N$ and $W$ are invertible in $\Hinf$, that $\sigma > 0$ and that (\ref{eq:FGepsilon}) holds. Then the optimization problem
  \begin{equation}\label{theoobj}
    \minimize_{C,D\in\Htwo} J(C,D)
  \end{equation}
  subject to
  \begin{equation}\label{theoconstr}
    \norm{CS}_2^2 + \norm{CM}_2^2 \leq \sigma^2
  \end{equation}
  attains a minimum value that is equal to the minimum of the convex optimization problem
  \begin{align}\label{Kcost}
    \minimize_{K\in\Htwo} \varphi(K),
  \end{align}
  which is attained by a unique minimizer. 
  
  Moreover, suppose $K\in\Htwo$ is a solution to (\ref{Kcost}). If $K$ is not identically zero, then  $C$ and $D$ solve (\ref{theoobj}) subject to (\ref{theoconstr}) if and only if $C\in\Htwo$,
  $D=KC^{-1}\in\Htwo$ and
  \begin{equation}\label{theocondition}
    |C|^2 = \frac{\sigma^2}{\norm{WKN\begin{bmatrix}S & M\end{bmatrix}}_1}\frac{|WKN|}{\sqrt{|S|^2 + |M|^2 }} \text{ on } \T.
  \end{equation}
  If $K=0$, then the solution to (\ref{theoobj}) and (\ref{theoconstr}) is given by
  $D=0$ and any function $C\in\Htwo$ that satisfies (\ref{theoconstr}). 
\end{thm}
\begin{proof}
Define $H\in\Hinf$ according to Lemma~\ref{lem:Hdefine}.
Then (\ref{theoconstr}) is equivalent to $\norm{CH}_2^2 \leq \sigma^2$.  
Define the sets
  \begin{align*}
    \Theta &= \left\{(C,D) : C,D\in{\Htwo},\
       \norm{CH}_2^2 \leq \sigma^2 \right\} \\
    \Theta(K) &= \left\{(C,D) : (C,D)\in \Theta,\ K=DC\right\}.
  \end{align*}
  Then the infimum of $J(C,D)$ subject to (\ref{theoconstr}) can be
  written
  \begin{align}
    &\inf_{C,D\in \Theta} J(C,D) \notag \\
    &= \inf_{K\in\Hone}
    \inf_{C,D\in\Theta(K)} J(C,D)\notag \\
    &= \inf_{K\in\Hone} \!\! \left( \! \norm{W(P-K)S}_2^2 + \! \norm{WKM}_2^2 + \!\!\!\!
      \inf_{C,D\in\Theta(K)} \!\! \norm{WDN}_2^2 \right)\notag\\
    &= \inf_{K\in\Hone} \norm{W(P-K)S}_2^2 +\norm{WKM}_2^2 +
    \frac{1}{\sigma^2} \norm{WKHN}_1^2 \notag \\
    &= \inf_{K\in\Hone} \varphi(K) \label{phiK}
  \end{align}
  The first equality is true by Theorem 17.10 in \cite{rudin86real}.
  The second equality follows because the first two terms in $\inf_{C,D\in\Theta(K)} J(C,D)$ are constant. The third equality follows from application of Lemma \ref{lem:factorization_coding} to perform the inner minimization. The final equality follows from (\ref{eq:HFG}).
  
  It will now be shown that the minimum is attained in (\ref{phiK}) by a
  unique $K\in\Htwo$. Completion of squares gives that
  \begin{align*}
    \varphi(K) &= \norm{W(P-K)S}_2^2 + \norm{WKM}_2^2 + \frac{1}{\sigma^2}\norm{WKHN}_1^2\\
    &=\norm{WPS}_2^2 + \norm{WKH}_2^2 \\ &- 2 \re \langle WPSS^*,WKHH^{-1} \rangle +
    \frac{1}{\sigma^2} \norm{WKHN}_1^2 \\
    &=\norm{WPSS^*H^{-*} - WKH}_2^2 + \frac{1}{\sigma^2}\norm{WKHN}_1^2 + \eta,
  \end{align*}
  where $\eta$ is a constant that does not depend on $K$.
  Let $X=WKH$ and $R=WPSS^*H^{-*}\in\Linf$.
  Minimizing $\varphi(K)$ over $K\in\Hone$ is then equivalent to minimizing
  \begin{align}\label{eq:coding_rho}
    \psi(X) = \norm{R-X}_2^2+\frac{1}{\sigma^2}\norm{XN}_1^2
  \end{align}
  over $X\in\Hone$. In the latter problem, it is sufficient to
  consider $X$ such that $\psi(X) \le \psi(0) = \norm{R}_2^2$. That is, only $X$ satisfying
  \begin{align*}
    \norm{X}_2 &= \norm{R-X-R}_2 
    \le \norm{R-X}_2 + \norm{R}_2 \\&\le \sqrt{\psi(X)}+\norm{R}_2 \le 2\norm{R}_2 \eqdef r.
  \end{align*}
  Now, in the weak topology, $\psi(X)$ is lower
  semicontinuous on $\Ltwo$ and the set $\left\{X:\norm{X}_2\le
    r\right\}$ is compact. This proves the existence of a minimum.
  The minimum is unique since $\psi(X)$ is strictly convex.
  Moreover, since $\norm{X}_2 \le r$, it is sufficient to minimize over $X\in\Htwo$ instead of $\Hone$.
  
  Suppose now that $X\in\Htwo$ minimizes $\psi(X)$. From $H^{-1},W^{-1}\in\Hinf$ it follows that
  $K=W^{-1}XH^{-1}\in\Htwo$ attains the infimum value in (\ref{phiK}) and that this value is equal to the minimum of (\ref{Kcost}). Since the minimum is attained in (\ref{phiK}) and, by Lemma~\ref{lem:factorization_coding}, there exists $(C,D)\in\Theta$ such that $J(C,D)=\varphi(K)$, it follows that the minimum of (\ref{theoobj}) subject to (\ref{theoconstr}) is attained.
  
The optimality condition (\ref{theocondition}) follows from the application of Lemma~\ref{lem:factorization_coding}, using that
  $|H|=\sqrt{|S|^2+|M|^2}$.
\end{proof}

\begin{rem}
  $\varphi(K)$ is convex, and $\varphi(\overline{K}) = \varphi(K)$. Thus, \[\varphi\left (\frac{K+\overline{K}}{2}\right ) \leq \frac{1}{2}\left (\varphi(K)+\varphi(\overline{K})\right ) = \varphi(K).\]
  Since the optimal $K$ is unique, this shows that the minimizing $K$ satisfies $K(e^{-i\omega}) = \overline{K(e^{i\omega})}$. Thus, $C$ can be chosen to have this property as well, meaning that $C$ can be approximated by a rational function with real coefficients. The same holds for $D$.
\end{rem}

\begin{rem}
  It was noted in Remark \ref{rem:coding_CD_nonunique} that the optimal factorization problem can have multiple solutions. To clarify, the optimal $K$ is unique but there are multiple factorizations of $K$ into $C$ and $D$ that achieve the minimum value of $J(C,D)$.
\end{rem}

It is noted that the solution of the problem essentially amounts to minimizing the sum of a 2-norm and a 1-norm of the decision variable. The 2-norm represents the cost in the Wiener filter problem, and the 1-norm represents the contribution of the channel noise to the error variance. The SNR $\sigma^2$ determines the relative importance of the two terms. For small SNR, the optimal $K$ will have small magnitude since the channel noise dominates the transmitted signal. As the SNR becomes larger, the magnitude of $K$ will become larger, and it will approach the Wiener filter in the limit when the SNR goes to infinity.

\subsection{Vector case}
In this section, the results in the previous section will be generalized to the case of vector-valued signals.

Consider again the system in Figure \ref{fig:coding_block_diagram} and assume that all signals are vector-valued and all systems are given by their corresponding transfer matrices. The number of elements in signal $s$ is denoted $n_s$ and so forth. That is, $s(k)\in\mathbb{R}^{n_s}$
Matrix dimensions are not explicitly stated in this section except when necessary. It is generally assumed that all matrices are of appropriate size.
In addition to all the assumptions made in the scalar case, it is now also assumed that:
\begin{enumerate}
\item The communication channel consists of $n_t$ parallel AWGN channels. The power constraint  (\ref{eq:scalar_power_constraint_time_domain}) is replaced by the total power constraint
  \begin{equation*}
    \E{t(k)^T t(k)} \leq \sigma^2.
  \end{equation*}  
\item All input signals in Fig. \ref{fig:coding_block_diagram} have identity covariance matrices. Moreover, 
$N(z)=W(z)=I$. That is, the channel noise is white with identity covariance and the frequency weight is uniform.
\item The number of elements in the signals satisfy 
  \begin{equation}\label{eq:coding_vector_lengths_constraint}
    n_t \geq \min\lbrace n_s, n_e\rbrace, 
  \end{equation}
  where $C$ is $n_t \times n_s$ and $D$ is $n_e \times n_t$.
  If the number of channels $n_t$ would be smaller than $n_f$ and $n_e$, then the product $DC$ could not have full rank. This means that optimization over $K=DC$ would have to include a rank constraint, which is very difficult to handle even in the static case.
\item The inequality (\ref{eq:FGepsilon}) is replaced by the matrix version
  \begin{align}\label{eq:FGepsilonMIMO}
    \exists \varepsilon > 0 \text{ such that } FF^*+GG^*\succeq \varepsilon I \text{ on }\T.
  \end{align}
\end{enumerate}

The objective is thus to minimize
\begin{equation*}
  J_v(C,D) = \norm{(P-DC)S}_2^2 + \norm{DCM}_2^2 + \norm{D}_2^2
\end{equation*} 
subject to 
\begin{equation}\label{eqcmSNRConstraint}
  \norm{CS}_2^2 + \norm{CM}_2^2 \leq \sigma^2
\end{equation}
The objective and the constraint are thus quite similar to the ones in the scalar case. It will be seen that the equivalent convex problem looks the same but that the optimality condition will, however, be more complicated. The main difference between the scalar and vector versions of the problem is that the optimal factorization (Lemma \ref{lem:factorization_coding}) is much more difficult to prove in the vector case.
\begin{lem}[Optimal factorization, vector case]\label{lem:factorization_MIMO}
  Suppose that $\sigma > 0$, $K \in \Hone$, that $H \in \Hinf$ is invertible in
  $\Hinf$ and that (\ref{eq:coding_vector_lengths_constraint}) holds. Then the optimization problem
  \begin{equation*}
    \minimize_{C,D \in \Htwo} \norm{D}_2^2
  \end{equation*}
  subject to
  \begin{equation*}
    K=DC, \quad \norm{CH}_2^2 \leq \sigma^2
  \end{equation*}
  attains the minimum value $\frac{1}{\sigma^2}\norm{KH}_1^2$.
  
  Moreover, suppose that $K$ is not identically zero and let \mbox{$K=K_i K_o$} be an inner-outer
  factorization and $K_o H = U_o \Sigma V^*$ be a singular value decomposition. Then $C,D \in \Htwo$ are optimal if and only if
  \begin{gather}
    K=DC,\quad \norm{CH}_2^2=\sigma^2,\\ \quad D D^* = \frac{\norm{KH}_1}{\sigma^2}
    K_i U_o\Sigma U_o^* K_i^*.
  \end{gather}

  If $K=0$ then the minimum is achieved by \mbox{$D=0$} and any
  function $C \in \Htwo$ that satisfies $\norm{CH}_2^2 \leq \sigma^2$.
\end{lem}
\begin{proof}
  If $K=0$ the proof is trivial, so assume that $K$ is not identically
  zero. Then neither $C$ nor $D$ are identically zero and
  $\alpha = \norm{CH}_2 > 0$. Now, suppose that
  $C, D$ are feasible and that $\alpha < \sigma$. Then
  \begin{align*}
    C_\alpha=\frac{\sigma}{\alpha}C, \quad D_\alpha =
    \frac{\alpha}{\sigma} D
  \end{align*} are feasible and $\norm{D_\alpha}_2 <
  \norm{D}_2$. Hence, a necessary condition for optimality is that
  \mbox{$\norm{CH}_2^2 = \sigma^2$}.

  The remainder of this proof is divided into three parts. First, the dual problem is considered. Then, it is shown that there is a saddle point and the optimality criteria are derived. Finally, existence of the solution is proven by construction. 

  \proofheader{Dual Problem:}
  In order to avoid dealing with analyticity constraints associated with $\Htwo$, the search will temporarily be relaxed to \mbox{$C, D \in \Ltwo$}. Later, it will be shown that there are \mbox{$C, D \in \Htwo$} that satisfy the derived optimality criteria. For $\lambda \geq 0$ and matrix-valued $\Phi \in \Linf$, introduce the Lagrangian
  \begin{align} 
    &L(C,D,\lambda,\Phi) = \norm{D}_2^2 + \lambda\left(\norm{CH}_2^2-\sigma^2\right) \notag \\
    &- \langle\re \Phi, \re DC-K\rangle \notag - \langle\im \Phi, \im DC-K\rangle \notag \\
    &= \norm{D}_2^2 + \lambda\left(\norm{CH}_2^2-\sigma^2\right) -
    \re\langle \Phi,
    DC-K\rangle \notag \\
    &= \! \intT \!\! \norm{D}_F^2 +
    \lambda\norm{CH}_F^2 - \re \tr  (\Phi^*(DC-K))
    \frac{d\omega}{2\pi} - \lambda\sigma^2 \label{eq:Lagrangian_integral}
  \end{align}
  The integrand in (\ref{eq:Lagrangian_integral}) can be rewritten as
  \begin{align}
    &{\norm{D}_F^2 + \lambda\norm{CH}_F^2
      - \re \tr \left(C\Phi^*D-\Phi^*K\right)}\notag \\
    &= \norm{D-\half\Phi C^*}^2_F \!\! + \lambda \norm{CH}_F^2 - \quarter \norm{C\Phi^*}_F^2 + \re \tr
    (\Phi^* K) \notag \\
    &=\norm{D-\half\Phi C^*}^2_F \!\!\! + \! \tr \! \left[C\!\left(\!\lambda H H^* - \quarter \Phi^*
          \Phi\!\right)\! C^* \! + \re \Phi^* K \right] 
            \label{eq:Lagrangian_integrand}
  \end{align}
  Only the first term depends on $D$. The contribution of this term is minimized by
  \begin{align}\label{eq:Dopt}
    D=\half \Phi C^*.
  \end{align}
  If (\ref{eq:Dopt}) holds, then $L$ only depends on $C$ through the first term inside the 
  brackets in (\ref{eq:Lagrangian_integrand}). Pointwise minimization of that term gives
  \begin{multline*}
    \inf_{C\in\Ltwo}\ \tr \left[C\left(\lambda H H^* - \quarter \Phi^*
        \Phi\right)C^*\right] \\ =
    \begin{cases}
      0, & 4\lambda H H^* \geq \Phi^* \Phi \text{ on } \T \\
      -\infty, & \text{otherwise.}
    \end{cases}
  \end{multline*}
  Moreover, the remaining term in (\ref{eq:Lagrangian_integrand}) can be written
  \begin{align*}
  \tr \left(\Phi^*K\right) = \tr \left(\Phi^*DC\right) = \half
    \tr\left(C\Phi^*\Phi C^*\right) = \half \norm{\Phi C^* }_F^2.
  \end{align*}
  Thus, $\tr \left(\Phi^*K\right)$ is real and non-negative, and
  \begin{align*}
    \inf_{C,D \in \Ltwo} \!\! L = \! 
    \begin{cases}
      \intT {\tr (\Phi^* K) \frac{d\omega}{2\pi}}
      -\lambda\sigma^2, & \!\! 4\lambda H H^* \! \geq \! \Phi^* \Phi \text{ on } \T \\ 
      -\infty, & \!\! \text{otherwise.}
    \end{cases}
  \end{align*}
  
  Introduce 
  \begin{equation*}
    \Psi = \frac{1}{2\sqrt{\lambda}} \Phi H^{-*}.
  \end{equation*}
  Then the dual problem can be written as
  \begin{align*}
    \maximize_{\lambda \geq 0, \Psi \in \Linf}\ 2\sqrt{\lambda}
    \intT {\tr \left(\Psi^* KH\right) \frac{d\omega}{2\pi}} -\lambda\sigma^2
  \end{align*}
  subject to 
  \begin{align}\label{eq:dualconstraint}
    \Psi^* \Psi \leq I \text{ on } \T.
  \end{align}

  The dual function is concave in $\lambda$. Letting $\lambda=0$ gives the value $0$. Since $\tr(\Psi^*KH)\geq0$ there exists $\lambda>0$ that gives a positive value, so the optimal $\lambda$ is given by the first-order condition
  \begin{align*}
    \left (\frac{1}{\sigma^2} \intT {\tr \left(\Psi^* KH\right) \frac{d\omega}{2\pi}}\right )^2 = \lambda,
  \end{align*}
  obtained by differentiation with respect to $\lambda$. With this $\lambda$ the dual problem simplifies to
  \begin{align}\label{eq:dualfunction}
    \maximize_{\Psi \in \Linf} \frac{1}{\sigma^2} \left (\intT {\tr \left(\Psi^* KH\right) \frac{d\omega}{2\pi}}\right )^2
  \end{align}
  subject to (\ref{eq:dualconstraint}).

  The integrand in (\ref{eq:dualfunction}) will now be maximized pointwise.
  Recall that $KH=K_i K_o H = K_i U_o \Sigma
  V^*$ and denote the number of rows of $K_o$ by $m$. Then $\Sigma$ is diagonal with diagonal elements $\sigma_k$, $k=1\ldots m$. Since $K$ is \mbox{$n_e\times n_s$} the rank of $K$ is not greater than $\min \lbrace n_e,n_s \rbrace$ and thus \begin{equation} \label{eq:coding_rankbound}
    m \leq \min \lbrace n_e,n_f \rbrace.
  \end{equation}
  
  $K_o$ is row outer by definition and $H$ is outer by Corollary 4.7 in \cite{Garnett}. It follows that $K_o H$ is row outer and thus has full row rank. It follows that the singular values are positive: \mbox{$\sigma_k >0$}, $k=1\ldots m$.
  Since $K_oH$ is wide (it has $n_s \geq m$ columns) it follows that $U_o$ is square and thus unitary.
  
  Define $U=K_i U_o$ and $\widetilde{\Psi} = U^* \Psi V$. Then it
  follows from (\ref{eq:dualconstraint}) and $UU^* \leq I$ that
  \begin{equation*}
    \widetilde{\Psi}^*\widetilde{\Psi} = V^*\Psi^*U U^* \Psi V \leq V^*\Psi^*
    \Psi V \leq V^* V =I.
  \end{equation*}
  Using $\widetilde{\Psi}$, an upper bound can be obtained for the integrand in (\ref{eq:dualfunction}):
  \begin{align*}
    \sup_{\Psi^* \Psi \leq I} \tr \left(\Psi^* KH\right) &=
    \sup_{\Psi^* \Psi \leq I} \tr \left(\Psi^* U\Sigma V^*\right) \\ 
    &= \sup_{\Psi^* \Psi \leq I} \tr \left(V^*\Psi^* U\Sigma \right) \\ 
    &\leq \sup_{\widetilde{\Psi}^* \widetilde{\Psi} \leq I} \tr \left(\widetilde{\Psi}^*\Sigma\right) \\
    &= \sum_{k=1}^m \sup_{\left|\widetilde{\Psi}_{kk}\right|\leq 1} \sigma_k \widetilde{\Psi}_{kk}
    = \sum_{k=1}^m \sigma_k
  \end{align*}
  The supremum is achieved if and only if $\widetilde{\Psi} = I$. Therefore, the upper bound is achieved by $\Psi$ if and only if $U^* \Psi V = I$ and $\Psi^*\Psi \leq I$. The set of $\Psi$ satisfying these conditions can be parametrized as:
  \begin{align}
    \Psi &= UV^* + \Psi_0 = K_i U_o V^* + \Psi_0 \label{eq:Psi0def} \\
    I &\geq \Psi^* \Psi \label{eq:PsiPsi},
  \end{align}
  where $\Psi_0$ satisfies
  \begin{equation}\label{eq:Psi0a}
    0=U^* \Psi_0 V = U_o^* K_i^* \Psi_0 V.
  \end{equation}
  Pre-multiplying (\ref{eq:Psi0a}) with $U_o$ gives the equivalent condition
  \begin{align}\label{eq:Psi0b}
    K_i^* \Psi_0 V &=0.
  \end{align}
  
  Choosing, for example, $\Psi_0 = 0$ gives $\Psi=UV^*$, which attains the upper bound. Hence, the value of the dual problem is  
  \begin{align*}
    &\max_{\Psi^* \Psi \leq I} \frac{1}{\sigma^2} \left(\intT \tr \left(\Psi^* KH\right) \frac{d\omega}{2\pi}\right)^2 \\
    &= \frac{1}{\sigma^2} \left( \intT \tr \left(VU^* U \Sigma V^*\right) \frac{d\omega}{2\pi}\right)^2    = \frac{1}{\sigma^2} \norm{KH}_1^2.
  \end{align*}
  
  The maximizing dual variables are given by
  \begin{align}\label{eq:Phiopt}
    \Phi = 2\sqrt{\lambda} \Psi H^* = 2\sqrt{\lambda} (K_i U_o V^* + \Psi_0) H^*
  \end{align}
  where $\Psi_0$ is such that (\ref{eq:Psi0def}), (\ref{eq:PsiPsi}) and (\ref{eq:Psi0b})
  hold, and 
  \begin{align}\label{eq:lambdaopt}
    \lambda = \left(\frac{1}{\sigma^2} \norm{KH}_1\right)^2.
  \end{align}
  
  \proofheader{Saddle Point:} 
  It will now be shown that there is a saddle point, which implies that
  the duality gap is zero. 

  In the following, assume that (\ref{eq:Psi0def}), (\ref{eq:PsiPsi}), 
  (\ref{eq:Psi0b}), (\ref{eq:Phiopt}) and (\ref{eq:lambdaopt})
  hold. Then $\lambda$ and $\Phi$ are dual feasible. The point
  $(C,D,\lambda,\Phi)$ is a saddle point if and only if $C,D \in
  \Htwo$ are primal feasible,
  \begin{align}\label{eq:saddle1}
    \lambda\left(\norm{CH}_2^2-\sigma^2\right) = 0
  \end{align}
  and
  \begin{align}\label{eq:saddle2}
    L(C,D,\lambda,\Phi) = \inf_{\widehat{C},\widehat{D} \in \Htwo} L(\widehat{C},\widehat{D},\lambda,\Phi).
  \end{align}
  The saddle point conditions imply that $\norm{CH}_2 = \sigma$ since
  $\lambda > 0$ and that $D=\half \Phi C^*$ as it was seen earlier
  that this follows from minimization of the Lagrangian.

  Suppose that the saddle point conditions hold. Then $C,D$ satisfy $K=DC$ and $D=\half \Phi C^*$. Moreover,
  \begin{align*}
    D D^* &= \half D C \Phi^* =\half K\Phi^*  
    =\sqrt{\lambda} K_i  K_o H (V U_o^*K_i^*+\Psi_0^*)\\
    &= \sqrt{\lambda} (K_i U_o\Sigma U_o^* K_i^* + K_i U_o \Sigma V^* \Psi_0^*).
  \end{align*}
  Clearly, $DD^*$ and $K_i U_o\Sigma U_o^* K_i^*$ are Hermitian. Accordingly,
  \begin{equation*}
    A=K_i U_o \Sigma V^* \Psi_0^*
  \end{equation*}
  must be Hermitian. Now, by
  (\ref{eq:Psi0b}),
  \begin{gather*}
    A K_i = K_iU_o \Sigma V^* \Psi_0^* K_i = 0 \nonumber\\
    \Rightarrow 0 = AK_i = A^* K_i = \Psi_0 V \Sigma U_o^* K_i^*K_i =
    \Psi_0 V \Sigma U_o^*. \nonumber
  \end{gather*}
  Hence, $A=0$ and
  \begin{align}\label{eq:DD}
    D D^* = \sqrt{\lambda} K_i U_o\Sigma U_o^* K_i^* = \frac{\norm{KH}_1}{\sigma^2} K_i U_o\Sigma U_o^* K_i^*.
  \end{align}

  Suppose instead that $C,D \in \Htwo$ satisfy $K=DC$, $\norm{CH}_2 =
  \sigma$ and (\ref{eq:DD}). Then $C,D$ are primal feasible and
  (\ref{eq:saddle1}) is satisfied. Moreover, 
  \begin{align*}
    L(C,D,\lambda,\Phi)&=\norm{D}_2^2 =
    \sqrt{\lambda}\intT \tr
    \left( K_i U_o\Sigma U_o^* K_i^* \right) \frac{d\omega}{2\pi} \\
    &= \sqrt{\lambda}\intT \tr
    \left(\Sigma\right) \frac{d\omega}{2\pi} = \frac{1}{\sigma^2}\norm{KH}_1^2,
  \end{align*}
  so (\ref{eq:saddle2}) holds and thus the saddle point conditions are
  satisfied. Since these assumptions and the saddle point conditions
  imply each other, they are equivalent.

  To conclude, it has been shown that $(C,D,\lambda,\Phi)$ is a saddle point, which implies that $C,D\in \Htwo$ achieve the claimed minimum, if and only if $K=DC$, $\norm{CH}_2^2 = \sigma^2$ and (\ref{eq:DD}) holds.

  \proofheader{Existence of Solution:} 
  Define $B=\sqrt{\lambda} U_o\Sigma U_o^* \in \Lone$, which is
  Hermitian with real diagonal. Recall that $K_o H$ is row outer with singular values \mbox{$\sigma_k >0$}, $k=1\ldots m$. From this and \mbox{Lemma \ref{lem:nonsquare_outer}} it follows that $\log \sigma_k \in \Lone$. Since $U_o$ is unitary it also follows that $B$ is positive definite. Moreover, 
  \begin{align*}
    \log \det B = \frac{m}{2}\log \lambda + \sum_{k=1}^m \log \sigma_k \in \Lone
  \end{align*}
  Therefore, according to the theorem in \cite{Wiener1959}, there is
  an outer transfer matrix $D_o \in \Htwo$ such that \mbox{$B = D_o
    D_o^*$}. 
  Let $\widetilde{D}=K_i D_o \in \Hardy{2}$ and \mbox{$\widetilde{C}=D_o^{-1} K_o$}. Then 
  \begin{align*}
    \widetilde{C} &= D_o^{-1}K_o H H^{-1} 
    = D_o^{-1} U_o \Sigma V^* H^{-1}\\
    &= D_o^{-1} U_o \Sigma U_o^* U_o V^* H^{-1}
    = \frac{1}{\sqrt{\lambda}} D_o^* U_o V^* H^{-1} \in \Lebesgue{2}
  \end{align*}
  Since $D_o$ is outer it follows from Lemma \ref{lem:product_with_inverse_in_Hp} that $\widetilde{C}\in \Hardy{2}$. 

  It can now be verified that $\widetilde{C}$ and $\widetilde{D}$ satisfy the
  optimality conditions:
  \begin{align*}
    \widetilde{D} \widetilde{C} = K_i D_o D_o^{-1} K_o = K_i K_o = K,
  \end{align*}
  \begin{align*}
    \norm{\widetilde{C}H}_2^2 \! &= \norm{D_o^{-1} K_o H}_2^2 = 
    \! \intT \! \! \tr \left( H^* K_o^* D_o^{-*} D_o^{-1} K_o H \right)
    \frac{d\omega}{2\pi} \\ 
    &= \intT \tr \left( V \Sigma U_o^* B^{-1}
      U_o \Sigma V^* \right) \frac{d\omega}{2\pi} \\
    &= \frac{1}{\sqrt{\lambda}} \intT \tr \left(\Sigma \right) \frac{d\omega}{2\pi}
    = \sigma^2
  \end{align*}
  and
  \begin{align*}
    \widetilde{D}\widetilde{D}^* = K_i D_o D_o^* K_i^* = \sqrt{\lambda} K_i
    U_o\Sigma U_o^* K_i^*.
  \end{align*}

  If the rank of $K$ does not equal $n_t$, then $\widetilde{C}$ and $\widetilde{D}$ are not of
  the required dimensions. $\widetilde{C}$ is $m \times n_s$ and
  $\widetilde{D}$ is $n_e \times m$, where, by (\ref{eq:coding_vector_lengths_constraint}) and (\ref{eq:coding_rankbound}), $m \leq \min\{n_e,n_f\}
  \leq n_t$. It is required that $C$ is $n_t \times n_s$ and that $D$ is $n_e \times n_t$. To solve this problem, let
  \begin{align*}
    D=
    \begin{bmatrix}
      \widetilde{D} & 0_{n_e \times n_t-m}
    \end{bmatrix} \in \Htwo, \quad C =
    \begin{bmatrix}
      \widetilde{C} \\ 0_{n_t-m \times n_s}
    \end{bmatrix} \in \Htwo.
  \end{align*}
  Noting that $DC=\widetilde{D}\widetilde{C}=K$, that $\norm{CH}_2 = \norm{\widetilde{C}H}_2 $ and that $DD^* = \widetilde{D}\widetilde{D}^*$
  it is \emph{finally} concluded that $C,D$ are optimal. 
\end{proof}

Just as in the scalar case, the solution to the optimal factorization problem can be used to find an equivalent convex problem. This problem looks exactly the same both cases. The theorem for the vector case is now stated.
\begin{thm}\label{thm:coding_equivalence_MIMO}
  Suppose that $\sigma > 0$, $S, M, P \in \Hinf$ and that (\ref{eq:coding_vector_lengths_constraint}) and (\ref{eq:FGepsilonMIMO}) hold.
  Then the optimization problem
  \begin{equation}\label{theoobj_MIMO}
    \minimize_{C,D\in\Htwo} J(C,D)
  \end{equation}
  subject to
  \begin{equation}\label{theoconstr_MIMO}
    \norm{CS}_2^2 + \norm{CM}_2^2 \leq \sigma^2
  \end{equation}
  attains a minimum value that is equal to the minimum of the convex optimization problem
  \begin{align}\label{Kcost_MIMO}
    \minimize_{K\in\Htwo} \norm{(P-K)S}_2^2 + \norm{KM}_2^2 + \frac{1}{\sigma^2} \norm{K       \begin{bmatrix}S & M\end{bmatrix}}_1^2,
  \end{align}
  which is attained by a unique minimizer. 
  
  Moreover, suppose $K\in\Htwo$ is a solution to (\ref{Kcost_MIMO}). If $K$ is not identically zero, then  $C, D\in \Htwo$ solve (\ref{theoobj_MIMO}) subject to (\ref{theoconstr_MIMO}) if and only if   \begin{gather*}
    K=DC,\quad \norm{C
      \begin{bmatrix}
        S & M
      \end{bmatrix}
    }_2^2=\sigma^2, \\
    D D^* = \frac{\norm{K
        \begin{bmatrix}
         S & M
        \end{bmatrix}
      }_1}{\sigma^2}  K_i U_o\Sigma U_o^* K_i^*,
  \end{gather*}
  where $K_i$ is defined by an inner-outer factorization \mbox{$K=K_i K_o$} and
  $U_o$ and $\Sigma$ are given by a singular value decomposition \mbox{$K_o H = U_o
    \Sigma V^*$}, where $H \in \Hinf$ satisfies $H^{-1} \in \Hinf$ and
  \mbox{$HH^*=SS^* + MM^*$}. 

  If $K=0$, then the solution to (\ref{theoobj_MIMO}) and (\ref{theoconstr_MIMO}) is given by
  $D=0$ and any function $C\in\Htwo$ that satisfies (\ref{theoconstr_MIMO}). 
\end{thm}
\begin{proof}
  With the assumption (\ref{eq:FGepsilonMIMO}), Lemma \ref{lem:Hdefine} holds in the matrix case as well. The rest of the proof is identical to the proof of Theorem \ref{thm:coding_equivalence_SISO}, except that Lemma \ref{lem:factorization_MIMO} is used instead of Lemma \ref{lem:factorization_coding}, with the obvious implications for the optimality conditions.
\end{proof}

\begin{rem}
  The assumption (\ref{eq:coding_vector_lengths_constraint}) may deserve some explanation. If
  there are too few communication channels relative to the dimensionality of $s$ and $e$, the maximum rank of the product $DC$ may be smaller than the smallest dimension of $K$. Then not all $K$ would be realizable as a product of $D$ and $C$, and a rank
  condition would have to be imposed on $K$ in Theorem \ref{thm:coding_equivalence_MIMO}. In principle, this changes nothing, but the assumption is included in order to avoid formulating the solution in terms of an optimization problems that cannot be reliably solved.
\end{rem}

\subsection{Optimal LTI Filters Require Infinite Memory}
The structure of optimal linear encoders and decoders will now be studied. In particular, it will be shown that the optimal filters generally have non-rational transfer functions. This corresponds to systems with infinite memory, since it is generally impossible to find a finite dimensional state-space realization of such transfer functions.

We consider the scalar case with white channel noise and rational $S,M,P$ and $W$. This implies that $N(z)=1$ and that $R=WPSS^*H^{-*}$ is rational, where $H$ satisfies (\ref{eq:HFG}). Since $S,M,P$ and $W$ are proper, it can safely be assumed that $R$ is proper. If $R$ is not proper then it can be made proper by multiplying $H$ with $z^{-k}$, for a large enough $k$. 

If we define
\begin{equation}
\psi(X) = \norm{R-X}_2^2+\frac{1}{\sigma^2}\norm{X}_1^2,
\end{equation}
the solution is given by solving the problem
 \begin{equation}
   \minimize_{X\in\Htwo} \psi(X).
   \label{eq:minsum}
 \end{equation}
Recall that the minimum of (\ref{eq:minsum}) is attained and that it is a strictly convex problem. It will now be shown that a necessary condition for the minimum cannot be satisfied by a rational $X$ except in some special cases. To begin with, two simple observations are made:
 
 \begin{enumerate}
 \item
 If the solution $X$ to (\ref{eq:minsum}) is a rational function it can be factorized into inner-outer factors as $X=FX_o$. The outer factor $X_o$ is then a rational function that solves the optimization problem
 \begin{equation}
 \minimize_{X_o\in\Htwo}\norm{F^*R-X_o}_2^2+\frac{1}{\sigma^2}\norm{X_o}_1^2
 \end{equation}
 where $F^*R$ is a rational function. Thus, we can assume without loss of generality that the optimal solution $X$ is outer.
\item Due to the orthogonality
 \begin{equation}
 \norm{R-X}_2^2=\norm{R_-}_2+\norm{R_+-X}_2^2,
 \end{equation}
 where $R=R_++R_-$ is a decomposition of $R$ into the analytical and anti-analytical parts, respectively, we can also assume that $R$ is analytical since the anti-analytical part does not affect the optimization. That is, $R=R_+$.
 \end{enumerate}

  Another assumption we make to simplify the proof is that the function $R$ has only simple poles. Note that the poles of $P_+(F^*R)$ are the same as of $P_+R$, so the simplicity of the poles remains true through the two rewritings above.
 
 \begin{thm}
   Consider the problem (\ref{eq:minsum}) with a proper non-constant rational function $R\in\Htwo$ and assume that the poles of $R$ are simple and that the optimal solution $X$ is not identically zero. Then $X$ is not a rational function.
 \end{thm}
 \begin{proof}
   We split the proof into several steps to underline the structure.
   
   \proofheader{Step 1}: Calculate the first variation of the functional $\psi$ and state the Euler-Lagrange equation. The standard differentiation of $\psi(X+\epsilon h)$ with respect to $\epsilon$ and then setting $\epsilon=0$ gives
   $$
   \delta\psi(h)=2\re\intT\biggl(\frac{\norm{X}_1}{\sigma^2} \frac{X}{|X|}-(R-X)\biggr)^*h\,\frac{d\omega}{2\pi}
   $$
For convex problems the necessary and sufficient condition for the minimum is that $\delta\psi(h)=0$ for all $h\in\Htwo$. It gives the Euler-Lagrange equation for the optimal $X$ as
   $$
   \frac{\norm{X}_1}{\sigma^2} P_+\frac{X}{|X|}=R-X
   $$
   where $P_+$ is the standard orthogonal projection from $\Ltwo$ to $\Htwo$. Note that the constant $\norm{X}_1\sigma^{-2}$ can be incorporated into $X$ and $R$ without affecting their rationality. So in the following we assume without loss of generality that this constant is equal to $1$ and analyze the equation
   \begin{equation}
     P_+\frac{X}{|X|}=R-X.
     \label{eq:EL}
   \end{equation}
   Since $X$ is not identically zero, it is not zero almost everywhere on $\T$ and the fraction $\frac{X}{|X|}$ is well defined.
   
   \proofheader{Step 2}: We will now assume that the solution $X\in\Htwo$ to (\ref{eq:EL}) is rational and show that it will lead us to a contradiction. In this step we prove that rationality of $R$ and $X$ in (\ref{eq:EL}) implies rationality of $|X|$. Indeed
   $$
   |X|=X^*\frac{X}{|X|}=X^*P_+\frac{X}{|X|}+X^*P_-\frac{X}{|X|}.
   $$
   The second term in the right hand side is anti-analytical, hence
\begin{equation}
   P_+|X|=P_+\bigl(X^*P_+\frac{X}{|X|}\bigr). \label{eq:projection_of_abs_X}
\end{equation}
   Clearly $P_+X|X|^{-1}$ is rational due to (\ref{eq:EL}) and thus the right hand side of (\ref{eq:projection_of_abs_X})   
   is rational too. Accordingly, $P_+|X|$ is also rational. Furthermore, the function $|X|$ is real and has a symmetric Laurent series. Therefore, the function $|X|$ must be rational itself. 
   
   Factorization as
   $|X|=h^*h=|h|^2=|h^2|$ with an outer rational $h\in\Htwo$ and assuming, as was explained previously, that $X$ is outer, gives the only possibility that $X=h^2$. That is, the rational solution $X$ must be a square of a rational function.
   
   \proofheader{Step 3}: Rewrite the Euler-Lagrange equation in terms of $h$ and then in terms of numerators and denominators of $h$ and $R$. Substituting $X=h^2$ into (\ref{eq:EL}) gives
   $$
   P_+\frac{X}{|X|}=P_+\frac{h^2}{h^*h}=P_+\frac{h}{h^*}=R-h^2.
   $$
   Introduce the notations for the numerators and the denominators
   $$
   h=\frac{p}{q},\quad R=\frac{b}{a}
   $$
   where $a$, $b$, $p$ and $q$ are polynomials. The polynomials $a$, $p$ and $q$ are stable by definition, since $R\in\Htwo$ and $h$ is outer in $\Htwo$. Introduce the notation for the conjugate polynomial to $p$ as
   $$
   \tilde p(z)=z^n p(z^{-1})
   $$
   where $n$ is the degree of $p$. The conjugate of a stable polynomial has the same degree and is anti-stable.
   With these notations in mind the Euler-Lagrange equation becomes
   \begin{equation}
     P_+\frac{p\tilde q z^{n-m}}{q\tilde p}=\frac{b}{a}-\frac{p^2}{q^2}=\frac{bq^2-ap^2}{aq^2}.
     \label{eq:ELnd}
   \end{equation}
   Here $n$ and $m$ are degrees of $p$ and $q$ respectively.
   
   \proofheader{Step 4}: Calculate the projection in the left hand side of (\ref{eq:ELnd}) and state the polynomial version of the Euler-Lagrange equation. We assume now that $n-m\ge 0$ and cover the opposite in the next step. Perform the partial fraction decomposition
   $$
   \frac{p\tilde q z^{n-m}}{q\tilde p}=\frac{Q}{q}+\frac{r}{\tilde p}
   $$
   where $Q$ is a polynomial and the degree of $r$ is less than $n$. Then
   $$
   P_+\frac{p\tilde q z^{n-m}}{q\tilde p}=\frac{Q}{q}
   $$
   and the equation (\ref{eq:ELnd}) becomes
   $$
   aqQ=bq^2-ap^2.
   $$
   Clearly $q(z)=0$ implies $a(z)=0$ since $p$ and $q$ are prime, hence $a=qa_0$ where $a_0$ is a polynomial. Canceling
   $q$ above we get
   $$
   a_0qQ=bq-a_0p^2.
   $$
   Similarly $q(z)=0$ implies $a_0(z)=0$ and thus $a_0=qa_1$. Canceling again gives
   $$
   a_1qQ=b-a_1p^2.
   $$
   Now it is clear that $a_1=1$ since otherwise $a_1(z)=0$ would give $b(z)=0$, which is impossible since $a$ and $b$ are also prime. Finally, we have $a=q^2$, which contradicts the assumption that zeros of $a$ are simple unless $q=a=1$. But for a proper non-constant $R$ it is impossible.
   
   \proofheader{Step 5}: The case $n-m<0$ is similar. Denote $k=m-n$. The only difference is in the partial fraction decomposition 
   $$
   \frac{p\tilde q}{q\tilde p z^k}=\frac{Q}{q}+\frac{r}{\tilde pz^k}
   $$
   where $Q$ is a polynomial and the degree of $r$ is less than $n+k=m$. The rest is exactly the same as in Step~4 with the same conclusion that $a=q^2$ which contradicts the assumption.
 \end{proof}
 
Because $S,M$ and $W$ are assumed to be rational and $X=WKH$ it follows that $K$ is rational if and only if $X$ is rational. Clearly, if $K$ is not rational, it cannot be factorized as $K=DC$ with rational $C$ and $D$. Thus, the transfer functions of optimal LTI encoders and decoders are not rational. 

As explained previously, this means that the filters can not be realized using finite memory. Obviously, approximations have to be done for practical implementation. For example, impulse responses of the filters may be truncated. It remains to investigate the impact on the performance of such approximations.

If the channel has noise-free feedback, that is, if $C$ has access to the channel output, then $C$ can estimate the states of $D$ exactly. It would be interesting to study if the memory of optimal linear encoders and decoders could be bounded in this case. Such a result would also be in line with the structural result for causal coders in \cite{Witsenhausen1979}, where the memory was bounded given that the encoder has knowledge of the decoder state.

\section{Numerical Solution}\label{sec:numerics}
A procedure for obtaining an approximate numerical solution will now be outlined for the vector version of the problem. 
\begin{enumerate}
\item The first step is to solve the optimization problem (\ref{Kcost_MIMO}) or, alternatively, minimize (\ref{eq:coding_rho}) (the constant part $\eta$ must then be added to obtain the distortion). An approximate solution can be obtained by using a finite basis representation of $K$ and approximating the integrals by sums over a finite number of frequency grid points. Such an approximated problem can be cast as a quadratic program with second-order cone constraints. 
\item Perform a matrix spectral factorization of $SS^*+MM^*$ to obtain $H\in\Hinf$ with $H^{-1}\in\Hinf$. \item Perform an inner-outer factorization of $K$ to obtain $K_iK_o=K$.
\item Perform a singular value decomposition of $K_oH$ to obtain $U_o\Sigma V^*=K_oH$. 
\item Use a finite basis approximation $A(\omega)$ of $DD^*$, for example using the parametrization 
  \begin{equation*}
    A(\omega) =A_0 + \sum_{k=1}^{N_c} A_k \left(e^{ki\omega} + e^{-ki\omega}\right)
  \end{equation*}
  and fit $A(\omega)$ to 
  \begin{equation*}
    \frac{\norm{K \begin{bmatrix} S & M \end{bmatrix}}_1}{\sigma^2}  K_i U_o\Sigma U_o^* K_i^*,
  \end{equation*}
  by minimizing the deviation in some suitable norm.
\item Perform a spectral factorization of $A(\omega)$, choosing $D_o$ as the stable and outer spectral factor.
\item Let $D=K_iD_o$ and $C=D_o^{-1}K_o$.
\item If $C$ and $D$ are of incorrect size, add rows of zeros to $C$ and columns of zeros to $D$ until they are of correct size.
\end{enumerate}
In the scalar case, the procedure is simplified as follows: Step 2 and 6 requires only scalar spectral factorizations, step 3, 4 and 8 are skipped and step 5 consists of fitting $A(\omega)$ to 
 \begin{equation*}
    \frac{\norm{K \begin{bmatrix} S & M \end{bmatrix}}_1}{\sigma^2} |KH|.
  \end{equation*}

\subsection{Example}
The numerical solution is illustrated by the following example. Consider the problem with $S=1/(z-0.9)$, $M=0$, $W=N=1$ and $P=z^{-d}$. The functional $\psi(X)$, given by (\ref{eq:coding_rho}), was approximated by discretization of the integrals over $4000$ 
grid points, uniformly placed on the unit circle. $X$ was parametrized as an FIR filter with $60$ coefficients. The minimization was then carried out for different SNR levels $\sigma^2$ and delays $d$, using 
Matlab, Yalmip \cite{yalmip} and SeDuMi \cite{sedumi}. 
    
The resulting MSE distortion levels are displayed in Fig.~\ref{fig:performance} together with the OPTA for the case with no delay constraint, obtained from (\ref{eq:minimum_distortion}). It can be seen that for small SNR's, the distortion is very close to the lower bound. This is not surprising since for zero SNR, the minimum distortion is $\norm{WPS}_2^2=\norm{S}_2^2$ over any type of coding system. 
For medium SNR's, the distortion is lower for longer delays. The difference seems, however, to decrease when the SNR becomes larger. The gap to the OPTA seems to approach about a factor two for high SNR's, regardless of delay. This suggests that for this source, it is the linearity, rather than the delay constraint, that is the performance-limiting factor for high SNR levels.

\begin{figure}[tb]
  \centerline{\includegraphics[width=1\hsize]{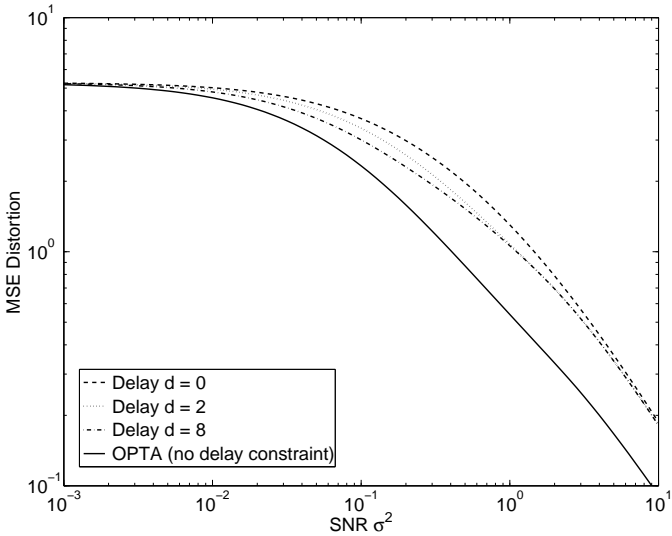}} 
  \caption{MSE distortion as a function of SNR level (logarithmic scale) for optimal linear encoders and decoders for three different delay constraints (approximate solutions), and the OPTA for the case without delay constraint. (Problem parameters: $S=1/(z-0.9)$, $M=0$, $W=N=1$, $P=z^{-d}$)}
  \label{fig:performance}
\end{figure}

\section{Conclusion}\label{sec:conclusion}
	This paper has shown how to find optimal LTI encoders and decoders for joint source-channel coding for Gaussian sources and channels. It has also been shown that such encoders and decoders in general require infinite memory. Thus, some approximation has to be done for numerical solution of the problem. It would be interesting to investigate if the performance loss due to such approximations can be somehow bounded.

In the scalar case, the solution has been extended to handle channels with feedback \cite{joh11phd}. This is not presented here to conserve space. Another extension is the problem of feedback control over AWGN channels, which will be the topic of an upcoming paper.

Possible topics for further research includes extending the solution in the MIMO case to channels with colored noise, investigating memory bounds when the channel has feedback and 
the suboptimality of linear solutions.

\appendix
\begin{lem}\label{lem:product_with_inverse_in_Hp}
  Suppose $Y \in \Smirnov$ is square and outer, $X \in \Smirnov$, and that \mbox{$Y^{-1}X \in \Lebesgue{p}$}. Then $Y^{-1}X \in \Hardy{p}$.
\end{lem}
\begin{proof}
  $Y^{-1} \in \Smirnov$ by Theorem 10 in \cite{Inouye}. It is easy to verify that the product of two $\Smirnov$ functions is $\Smirnov$. The proof follows from the fact that $\Lebesgue{p} \cap \Smirnov = \Hardy{p}$ \cite{Garnett}.
\end{proof}

\begin{lem}\label{lem:nonsquare_outer}
  Suppose that $m\leq n$ and that the $m \times n$ transfer matrix $X \in \Hardy{p}$, \mbox{$p\in\lbrace1,2,\infty\rbrace$}, is row outer. 
  Then the singular values of $X$ satisfy 
  \[\log \sigma_k \in \Lone, \quad k=1\ldots m.\]
\end{lem}
\begin{proof}
  By Theorem 8 in \cite{Inouye} there exists a factorization \mbox{$X = X_{co} X_{i}$}, where $X_{co}$ is column outer and $X_i$ is inner. Since $X_{co}$ has full column rank on $\T$ it cannot have more columns than rows, and since $X$ is row outer $X_{co}$ cannot have fewer rows than columns. Thus $X_{co}$ is
  $m\times m$ and hence, by Theorem 10 in \cite{Inouye}, 
  $\det X_{co}$ is outer and thus $\det X_{co} \in \Smirnov$. According to a statement in section 17.19 in \cite{rudin86real} it follows that \mbox{$\log \left| \det X_{co} \right| \in \Lone$}.

  For the singular values of $X$, it holds that
  \begin{align*}
    \sum_{k=1}^m \log \sigma_k &= \half \log \prod_{k=1}^m \sigma_k^2 =
    \half \log \det XX^* \\ &= \half \log \det X_{co} X_{i} X_{i}^*
    X_{co}^* = \half \log \det X_{co} X_{co}^* \\ &= \log \left| \det
      X_{co} \right| \in \Lone.
  \end{align*}
  Furthermore, $\sigma_k \in \Lone$ since $X\in\Hardy{p}$.
  Because $\log \sigma_k < \sigma_k$ it holds that
  \begin{align*}
    \intT \log \sigma_k\ d\omega< \intT \sigma_k
    \ d\omega < \infty, \quad k=1\ldots m
  \end{align*}
  Since the sum of the logarithms is $\Lone$ and every term has an
  integral bounded from above, it follows that the integral of every term also must be bounded from below. That is,
  \begin{align*}
    \intT \log \sigma_k\ d\omega > -\infty, \quad k=1\ldots m
  \end{align*}
  and hence $\log \sigma_k \in \Lone$, $k=1\ldots m$
\end{proof}

\section*{Acknowledgment}
The authors would like to thank John Doyle (California Institute of Technology) for suggesting the problem,  and Nuno Martins (University of Maryland) and Mikael Skoglund (Royal Institute of Technology) for technical discussions and helpful comments.

\ifCLASSOPTIONcaptionsoff
  \newpage
\fi


\begin{thebibliography}{10}
\providecommand{\url}[1]{#1}
\csname url@rmstyle\endcsname
\providecommand{\newblock}{\relax}
\providecommand{\bibinfo}[2]{#2}
\providecommand\BIBentrySTDinterwordspacing{\spaceskip=0pt\relax}
\providecommand\BIBentryALTinterwordstretchfactor{4}
\providecommand\BIBentryALTinterwordspacing{\spaceskip=\fontdimen2\font plus
\BIBentryALTinterwordstretchfactor\fontdimen3\font minus
  \fontdimen4\font\relax}
\providecommand\BIBforeignlanguage[2]{{%
\expandafter\ifx\csname l@#1\endcsname\relax
\typeout{** WARNING: IEEEtran.bst: No hyphenation pattern has been}%
\typeout{** loaded for the language `#1'. Using the pattern for}%
\typeout{** the default language instead.}%
\else
\language=\csname l@#1\endcsname
\fi
#2}}

\bibitem{joh11phd}
E.~Johannesson, ``Control and communication with signal-to-noise ratio
  constraints,'' Ph.D. dissertation, Department of Automatic Control, Lund
  University, Sweden, Oct. 2011.

\bibitem{pilc1969}
R.~J. Pilc, ``The optimum linear modulator for a gaussian source used with a
  gaussian channel,'' \emph{Bell System Technical Journal}, vol.~48, no.~9,
  Nov. 1969.

\bibitem{lee_petersen1976}
K.-H. Lee and D.~P. Petersen, ``Optimal linear coding for vector channels,''
  \emph{IEEE Transactions on Communications}, vol.~24, no.~12, pp. 1283 --
  1290, Dec. 1976.

\bibitem{basar_sankur_abut1980}
T.~Basar, B.~Sankur, and H.~Abut, ``Performance bounds and optimal linear
  coding for discrete-time multichannel communication systems (corresp.),''
  \emph{IEEE Transactions on Information Theory}, vol.~26, no.~2, pp. 212 --
  217, Mar. 1980.

\bibitem{kochman_zamir2011}
Y.~Kochman and R.~Zamir, ``Analog matching of colored sources to colored
  channels,'' \emph{IEEE Transactions on Information Theory}, vol.~57, no.~6,
  pp. 3180 --3195, June 2011.

\bibitem{Dobrushin1962}
R.~Dobrushin and B.~Tsybakov, ``Information transmission with additional
  noise,'' \emph{IRE Transactions on Information Theory}, vol.~8, no.~5, pp.
  293--304, Sept. 1962.

\bibitem{Wolf1970}
J.~Wolf and J.~Ziv, ``Transmission of noisy information to a noisy receiver
  with minimum distortion,'' \emph{IEEE Transactions on Information Theory},
  vol.~16, no.~4, pp. 406--411, July 1970.

\bibitem{berger1971rate}
T.~Berger, \emph{Rate distortion theory: a mathematical basis for data
  compression}.\hskip 1em plus 0.5em minus 0.4em\relax Englewood Cliffs, NJ,
  USA: Prentice-Hall, 1971.

\bibitem{Witsenhausen1980}
H.~Witsenhausen, ``Indirect rate distortion problems,'' \emph{IEEE Transactions
  on Information Theory}, vol.~26, no.~5, pp. 518--521, Sept. 1980.

\bibitem{Neuhoff1982}
D.~Neuhoff and R.~Gilbert, ``Causal source codes,'' \emph{IEEE Transactions on
  Information Theory}, vol.~28, no.~5, pp. 701--713, Sept. 1982.

\bibitem{Witsenhausen1979}
H.~Witsenhausen, ``On the structure of real-time source coders,'' \emph{Bell
  Syst. Tech. J.}, vol.~58, no.~6, pp. 1437--1451, July-Aug 1979.

\bibitem{Walrand}
J.~Walrand and P.~Varaiya, ``Optimal causal coding--decoding problems,''
  \emph{IEEE Transactions on Information Theory}, vol.~29, no.~6, pp. 814--820,
  Nov. 1983.

\bibitem{teneketzis2006}
D.~Teneketzis, ``On the structure of optimal real-time encoders and decoders in
  noisy communication,'' \emph{Information Theory, IEEE Transactions on},
  vol.~52, no.~9, pp. 4017 --4035, sept. 2006.

\bibitem{gastpar}
M.~Gastpar, ``To code or not to code,'' Ph.D. dissertation, EPFL, Lausanne,
  2002.

\bibitem{derpich2011}
M.~S. Derpich and J.~{\O}stergaard, ``Improved upper bounds to the causal
  quadratic rate-distortion function for gaussian stationary sources,''
  \emph{IEEE Transactions on Information Theory}, accepted for publication.

\bibitem{Derpich08}
M.~Derpich, E.~Silva, D.~Quevedo, and G.~Goodwin, ``On optimal perfect
  reconstruction feedback quantizers,'' \emph{IEEE Transactions on Signal
  Processing}, vol.~56, no.~8, pp. 3871--3890, Aug. 2008.

\bibitem{DerpichPhd}
M.~Derpich, ``Optimal source coding with signal transfer function
  constraints,'' Ph.D. dissertation, University of Newcastle, 2009.

\bibitem{Borkar2001}
V.~Borkar, S.~Mitter, and S.~Tatikonda, ``Optimal sequential vector
  quantization of markov sources,'' in \emph{Proc. IEEE Conference on Decision
  and Control}, vol.~1, 2001, pp. 205--210.

\bibitem{Yuksel2010}
\BIBentryALTinterwordspacing
S.~Y{\"u}ksel, ``On optimal causal coding of partially observed markov sources
  in single and multi-terminal settings,'' 2010. [Online]. Available:
  \url{http://arxiv.org/abs/1010.4824v2}
\BIBentrySTDinterwordspacing

\bibitem{Mahajan2008}
A.~Mahajan and D.~Teneketzis, ``On the design of globally optimal communication
  strategies for real-time noisy communication systems with noisy feedback,''
  \emph{IEEE Journal on Selected Areas in Communications}, vol.~26, no.~4, pp.
  580--595, May 2008.

\bibitem{mahajan09}
------, ``Optimal design of sequential real-time communication systems,''
  \emph{IEEE Transactions on Information Theory}, vol.~55, no.~11, pp.
  5317--5338, Nov. 2009.

\bibitem{Wiener1964}
N.~Wiener, \emph{Extrapolation, Interpolation, and Smoothing of Stationary Time
  Series}.\hskip 1em plus 0.5em minus 0.4em\relax The MIT Press, 1964.

\bibitem{martins07}
N.~Martins, M.~Dahleh, and J.~Doyle, ``Fundamental limitations of disturbance
  attenuation in the presence of side information,'' \emph{IEEE Transactions on
  Automatic Control}, vol.~52, no.~1, pp. 56--66, Jan. 2007.

\bibitem{Garnett}
J.~Garnett, \emph{Bounded analytic functions}, revised 1st~ed.\hskip 1em plus
  0.5em minus 0.4em\relax New York, NY, USA: Springer, 2007.

\bibitem{rudin86real}
W.~Rudin, \emph{Real and Complex Analysis}, 3rd~ed.\hskip 1em plus 0.5em minus
  0.4em\relax McGraw-Hill Science/Engineering/Math, May 1986.

\bibitem{Inouye}
Y.~Inouye, ``Linear systems with transfer functions of bounded type: Canonical
  factorization,'' \emph{IEEE Transactions on Circuits and Systems}, vol.~33,
  no.~6, pp. 581--589, June 1986.

\bibitem{Wiener1959}
N.~Wiener and E.~Akutowicz, ``A factorization of positive {H}ermitian
  matrices,'' \emph{Indiana Univ. Math. J.}, vol.~8, pp. 111--120, 1959.

\bibitem{yalmip}
J.~L\"{o}fberg, ``Yalmip : A toolbox for modeling and optimization in
  {MATLAB},'' in \emph{Proceedings of the CACSD Conference}, Taipei, Taiwan,
  2004.

\bibitem{sedumi}
J.~Sturm, ``Using {SeDuMi} 1.02, a {MATLAB} toolbox for optimization over
  symmetric cones,'' \emph{Optimization Methods and Software}, vol. 11--12, pp.
  625--653, 1999, version 1.3 available from
  {\texttt{http://sedumi.ie.lehigh.edu/}}.

\end{thebibliography}
\end{document}